\documentclass[11pt,twocolumn,reprint,superscriptaddress,floatfix,showpacs]{revtex4-2}
\usepackage[T1]{fontenc}
\usepackage{tgbonum}   
\usepackage[british]{babel}  
\usepackage[colorlinks=true,allcolors=blue]{hyperref}  
\usepackage{graphicx} 
\usepackage[babel]{microtype}  
\usepackage{amsmath,amsthm,bbm} 
\usepackage{setspace}
\usepackage{braket}
\usepackage{physics}
\usepackage{blochsphere}
\usepackage{pgfplots}
\usetikzlibrary{calc,3d,shapes, pgfplots.external, intersections}

\usepackage[bitstream-charter]{mathdesign}
\usepackage{tikz}
\usepackage{mathtools}
\usetikzlibrary{arrows.meta}
\usepackage{csquotes}

\usepackage{xspace}  
\usepackage{pgfplots}
\usepackage{verbatim}
\usepackage{mathpazo}

\newcommand\id{\leavevmode\hbox{\small1\kern-3.3pt\normalsize1}}

\newtheorem{theorem}{Theorem}
\newtheorem{lemma}{Lemma}

\newtheorem{definition}{Definition}

\allowdisplaybreaks
\begin{document}
\title{Broadcasting of Nonlocality}

\author{Dhrumil Patel}
\affiliation{Department of Computer Science, Cornell University, Ithaca, New York 14853, USA}
\email{djp265@cornell.edu}
\affiliation{Center for Security, Theory and Algorithmic Research, International Institute of Information Technology Hyderabad, Gachibowli, Hyderabad-500032, Telangana, India}

\author{Arup Roy}
\affiliation{Department of Physics, A B N Seal College, Cooch Behar, West Bengal 736 101, India}

\author{Indranil Chakrabarty}
\affiliation{Centre for Quantum Science and Technology, International Institute of Information Technology Hyderabad, Gachibowli, Hyderabad-500032, Telangana, India}
\affiliation{Center for Security, Theory and Algorithmic Research, International Institute of Information Technology Hyderabad,
Gachibowli, Hyderabad-500032, Telangana, India}

\author{Nirman Ganguly}
\affiliation{ Department of Mathematics, Birla Institute of Technology and Science Pilani, Hyderabad Campus, Telangana-500078, India.}

\begin{abstract}
Bell nonlocality and steering are archetypal characteristics of quantum mechanics that mark a significant departure from conventional classical notions. They basically refer to the presence of quantum correlations between separated systems which violate a nonlocal inequality, the violation otherwise not possible if we restrict ourselves only to classical correlations. In view of the importance of such unique correlations one may be interested to generate more states exhibiting nonlocality starting from a few, a protocol which is termed as broadcasting. However, in the present submission, we show using universal Buzek-Hillary(BH) quantum cloning machine that, if one restricts to broadcasting through local quantum cloning, then such nonlocality cannot be broadcasted. Our study is done in the purview of the Bell-CHSH inequality and the CJWR (E.G.Cavalcanti,S.J. Jones,H.M Wiseman and M.D. Reid, Phys.Rev.A 80,032112(2009)) steering inequality. It is observed that when number of measurement settings is greater than 6, some of the output states are steerable after the application of local optimal B-H cloners. We find that, under some restrictions, if the Werner and Bell diagonal states are subjected to such procedures, then the resultant state is rendered unsteerable. We extend this study to three qubit systems and find that genuine tripartite nonlocality cannot be broadcasted using universal BH local quantum cloners, when we consider the Svetlichny's inequality. For two qubit systems, we have considered $10^5$ simulated general local unitaries over Werner and Bell-diagonal states and find that for none of these states broadcasting on nonlocality and 3-steerability is possible.
\end{abstract}

\maketitle
\section{Introduction}
Bell's inequality \cite{bell1} provides a significant yardstick for the presence of nonlocality in a system. The fact that some quantum states violate Bell's inequality and are thus nonlocal, contributed significantly to our understanding of the foundational issues in quantum mechanics \cite{bell1,bell2}. A plethora of work also exploited nonlocality from an information processing perspective leading to works on quantum cryptography \cite{quantumkeydistribution1, quantumkeydistribution2}. Quantum steering \cite{Steeringreview}, a notion first envisaged by Schrodinger \cite{SchSteering} and later given an operational interpretation in \cite{SteeringIeq1,hySteering1}, is another form of nonlocal correlations and is considered as a  hybrid between entanglement and Bell nonlocality. Alike Bell nonlocality, various steering inequalities have been proposed to witness steerability in quantum systems \cite{SteeringIeq2,SteeringIeq3,SteeringIeq4,SteeringIeq5,SteeringIeq6,SteeringIeq7,SteeringIeq8,SteeringIeq9,SteeringIeq10,SteeringIeq11}. 

Given the importance of a quantum resource, a tactical scheme could be to generate more such resources through copying. However, the no-cloning theorem in quantum mechanics \cite{wootters} puts a restriction on perfect cloning. The no-cloning theorem states that there exists no quantum mechanical process that can take two different non-orthogonal quantum states $|\psi_1 \rangle$ and $|\psi_2\rangle$ to states $|\psi_1\rangle\otimes|\psi_1\rangle$ and $|\psi_2\rangle\otimes|\psi_2\rangle$, respectively. Even though we cannot copy 
 an unknown quantum state perfectly, quantum mechanics never rules out the possibility of cloning it approximately \cite{approximatecloning1, approximatecloning2,approximatecloning3, cloningbuzek,buzek}. It also allows probabilistic cloning as one can always clone an arbitrary quantum state perfectly with some non-zero probability of success \cite{approximatecloning2,approximatecloning3,statedependentmachine}.

Buzek and Hillary introduced the concept of approximate cloning with certain fidelity. In this process, the state independent quantum copying machine was introduced by keeping the fidelity of cloning independent of the input state parameters. This machine is popularly known as universal quantum cloning machine (UQCM) \cite{buzek} which was later proven to be optimal \cite{buzekoptimal,buzekoptimal1}. Apart from this, there is also state dependent quantum cloning machine for which the quality of the copies depends on the input state parameters \cite{statedependentmachine,manish}. For the purpose of our work, we refer to these quantum cloners as Buzek-Hillary (B-H) quantum cloners after the authors who introduced them.

\par On a different but close perspective, the term broadcasting can be used in different contexts. Classical theory permits broadcasting of information, however, that is not the case for all states in quantum theory. In this context, Barnum et al. were the first to show that non-commuting mixed states do not meet the criteria of broadcasting \cite{Barnum}. It is impossible to have a process which will perfectly broadcast an arbitrary quantum state \cite{Barnum}.  Interestingly, when we talk about broadcasting of correlations (resources in general), we refer to a situation  that the correlations in a two-qubit state $\rho_{AB}$ are  broadcastable with the help of local operations if there exist two operations, $\Sigma_A$: $S(\mathbb{H}_A) \rightarrow S(\mathbb{H}_{A_{1}}\otimes \mathbb{H}_{A_{2}})$ and $\Sigma_B$: $S(\mathbb{H}_B) \rightarrow S(\mathbb{H}_{B_{1}} \otimes \mathbb{H}_{B_{2}})$ such that $I(\rho_{A_1B_1})$ = $I(\rho_{A_2B_2})$ = $I(\rho_{AB})$. Here, $I(\rho_{AB})$ is a correlation measure and $\rho_{A_1A_2B_1B_2}:=\Sigma_{A} \otimes \Sigma_B (\rho_{AB})$ and $\rho_{A_iB_i}:=\operatorname{Tr}_{A_{\bar{j}}B_{\bar{j}}}(\rho_{A_1A_2B_1B_2})$, where $j \neq i$. The notation $S(\mathbb{H}_x) $ refers to the set of density matrices over Hilbert space $ \mathbb{H}_x  $. One of the possible examples of the operations   $\Sigma_A, \Sigma_B$ are quantum cloning operations. Using local and nonlocal cloning operations, entanglement was broadcasted, where the input state was a two-qubit system \cite{buzeketal,cloningbuzek,bandKar,Ighui,Chat,Ajain}, a qubit-qudit system \cite{Mundra}, or a qutrit-qutrit system \cite{Mundra21}. In addition, there has been a recent progress in broadcasting of other general resources, including coherence and correlation that goes beyond the notion of entanglement \cite{Chat,Ajain,Mundra,UKsharma}.
\par Entanglement is a weaker form of nonlocality compared to steering and Bell nonlocality. This is because there exists entangled states which are unsteerable and local. In this work, we explore the possibility of broadcasting steerable and Bell nonlocal states. We show that, contrary to the case of entanglement, it is in fact impossible to broadcast Bell nonlocality when we use BH local cloners as local operations. Furthermore, we also apply $10^5$ simulated general local unitaries to conjecture the impossibility of broadcasting nonlocality. This study is done under the purview of the Bell-CHSH inequality \cite{CHSH}. On the other hand, while checking for quantum steering, using the well-known CJWR inequality \cite{SteeringIeq3}, we find that when the number of observables is more than 6, broadcasting of steering is possible. However, we find that  Werner \cite{RFWerner} and Bell diagonal states become unsteerable for all measurement settings after they undergo the above BH local cloning transformation. We also extend our study to general three-qubit systems where we use the Svetlichny's inequality \cite{Svetlichny} to demonstrate that genuine tripartite nonlocality cannot be broadcasted using BH local quantum cloners. 
 \par The paper is organized as follows: In section \ref{prelim} we recapitulate the basic notions pertinent to our study, in section \ref{broad}, we present our results on broadcasting in two qubit systems followed by the result obtained in tripartite systems in section \ref{broadtri}. In section \ref{broadgen}, we report on our numerical results followed by the conclusion in section \ref{conc}.

\section{Preliminaries}\label{prelim}
\noindent In this section, we briefly revisit some basic concepts such as nonlocality of quantum states and quantum cloning. All these concepts will be useful and relevant for the main findings of our manuscript. The notations $ A, B, \text{ and }C $ represent the system labels, and when written in bold, they represent the respective measurement operators.
\subsection{General two- and three-qubit states }
\noindent To begin with, we recall the definition of a general two-qubit state shared between two parties Alice ($A$) and Bob ($B$). It's representation in the canonical form is given as,
\begin{eqnarray}
\rho_{AB}&:=&\frac{1}{4}\big[\mathbb{I}_4+\sum_{i=1}^{3}\big(x_{i}(\sigma_{i}\otimes \mathbb{I}_2)+ y_{i}(\mathbb{I}_2\otimes\sigma_{i})\big)\nonumber\\
&+&\sum_{i,j=1}^{3}t_{ij}(\sigma_{i}\otimes\sigma_{j})\big]=\{\vec{x},\:\vec{y},\: \mathbb{T}\}\:\:\: \mbox{(say),}
\label{eq:genmixedstate}
\end{eqnarray}
where $x_i=\operatorname{Tr}[\rho_{AB}(\sigma_{i}\otimes \mathbb{I}_2)]$ and $y_i=\operatorname{Tr}[\rho_{AB}(\mathbb{I}_2\otimes\sigma_{i})]$ are local Bloch vectors. The correlation matrix is given by $\mathbb{T}=[t_{ij}]$ where $t_{ij}=\operatorname{Tr}[\rho_{AB}(\sigma_i\otimes\sigma_{j})]$. Here [$\sigma_i;\:i$ = $\{1,2,3\}$] are $2\otimes 2$ Pauli matrices, and $\mathbb{I}_n$ is the identity matrix of dimension $n \times n$.

Similarly, we define the general form of a three-qubit state shared between three parties Alice ($A$), Bob ($B$), and Charlie ($C$) as,
\begin{align}
\rho_{ABC}&:= \frac{1}{8}\Bigg[\mathbb{I}_8 
+\sum_{i=1}^{3}\big(x_{i} (\sigma_{i}\otimes \mathbb{I}_2 \otimes \mathbb{I}_2) + y_{i}(\mathbb{I}_2\otimes\sigma_{i} \otimes \mathbb{I}_2)\nonumber\\
& \qquad + z_{i} (\mathbb{I}_2 \otimes \mathbb{I}_2 \otimes\sigma_{i}) \big)
+\sum_{i,j=1}^{3} \big( t_{ij} (\mathbb{I}_2 \otimes\sigma_{i}\otimes\sigma_{j}) \nonumber\\
& \qquad + v_{ij} (\sigma_{i} \otimes \mathbb{I}_2 \otimes\sigma_{j}) + w_{ij} (\sigma_{i}\otimes\sigma_{j} \otimes\mathbb{I}_2 )\big) \nonumber\\
& \qquad + \sum_{i,j,k=1}^{3} G_{ijk} (\sigma_{i}\otimes\sigma_{j}\otimes\sigma_{k}) \Bigg]\nonumber\\
&=\{\vec{x},\:\vec{y},\:\vec{z},\:\mathbb{T},\:\mathbb{V},\: \mathbb{W}, \:\mathbb{G}\}\:\:\: \mbox{(say),}
\label{eq:genmixedstate3}
\end{align}
where the coefficients can be obtained by tracing the matrix multiplication of $\rho_{ABC}$ and their corresponding terms. For example, $x_i = \operatorname{Tr}[\rho_{ABC}(\sigma_{i}\otimes \mathbb{I}_2 \otimes \mathbb{I}_2)]$, $t_{ij} = \operatorname{Tr}[\rho_{ABC}(\mathbb{I}_2 \otimes\sigma_{i}\otimes\sigma_{j})]$, and $G_{ijk} = \operatorname{Tr}[\rho_{ABC}(\sigma_{i}\otimes\sigma_{j}\otimes\sigma_{k})]$.

\subsection{Nonlocality of a Quantum State }\label{subsec:nonlocality}
\noindent Bell nonlocality is a phenomenon arising out of some measurements made on a composite system, contradicting the assumptions of locality and realism. Bell's inequality provides a convenient tool to detect nonlocality. Any state that violates this inequality is said to exhibit Bell nonlocality. Bell nonlocality, \textit{a-priori} has nothing to do with quantum mechanics, however  certain correlations arising from quantum states violate a suitably chosen Bell's inequality and thus underscore the presence of nonlocal correlations.
\par Quantum steering is the manifestation of the non-classical correlations obtained between the outcomes of measurements applied on one part of an entangled state and the  post-measurement state remaining with the other part. Although noted by Schrodinger in his seminal paper \cite{SchSteering}, the notion of steering gained prominence with an operational reformulation in \cite{hySteering1}. A steering test can be seen as an entanglement test where one of the parties performs measurements. Thus, steering can be considered as a form of nonlocality occupying a place between entanglement and Bell nonlocality \cite{hySteering1}.
\par A much weaker form of nonlocality is quantum entanglement. A quantum state is called entangled if it cannot be expressed as a convex combination of product states. The inability to write the state in the above form gives rise to the identification of entanglement which is one of the most important quantum resources.

We now put down formal mathematical definitions of the above features.
\par \textbf{Bell bipartite nonlocality:} Consider that there are two parties Alice and Bob. Let $\mathcal{O}_{\mathbf{A}}$ and $\mathcal{O}_{\mathbf{B}}$ represent the set of observables on Alice's and Bob's side, respectively. Then  we have
\begin{equation}
P(a,b \vert \mathbf{A},\mathbf{B}; \rho_{AB}) = \operatorname{Tr}\left [(\Pi_{a}^{\mathbf{A}} \otimes \Pi_{b}^{\mathbf{B}} ) \rho_{AB}\right],
\end{equation}
where $P(a,b \vert \mathbf{A},\mathbf{B}; \rho_{AB})$ represents the probability of obtaining outcomes $a$ and $b$. Here $a$ and $b$ are respective eigenvalues of $\mathbf{A} \in \mathcal{O}_{\mathbf{A}} ,\mathbf{B} \in \mathcal{O}_{\mathbf{B}} $ and $ \Pi_{a}^{\mathbf{A}}, \Pi_{b}^{\mathbf{B}} $ are the respective projectors.

Bell nonlocality is said to be exhibited by a state $\rho_{AB}$ if there is at least one probability which cannot be written in the form given below:
\begin{equation}\label{eq:bell-nonlocality-break}
P(a,b \vert \mathbf{A},\mathbf{B}; \rho_{AB}) =  \sum_{\xi} p(\xi) p(a \vert \mathbf{A} , \xi) p(b \vert \mathbf{B} , \xi),
\end{equation}
where $ p(a \vert \mathbf{A} , \xi)$ and $p(b \vert \mathbf{B} , \xi)$ are some probability distributions and $\xi$ is a local hidden variable. This essentially means that some correlations obtained from quantum states cannot be explained by a local hidden variable model, i.e., they cannot be expressed as in \eqref{eq:bell-nonlocality-break}.
\par In \cite{Horo1, Horo2}, Horodecki et al. found a closed-form pertaining to the violation of the Bell-CHSH inequality \cite{CHSH} optimized over all possible measurements. The maximum Bell-CHSH violation $V(\rho_{AB})$ for a two-qubit state $\rho_{AB}$ is given by 
\begin{equation}
V(\rho_{AB})=2\sqrt{M(\rho_{AB})},
\end{equation}
where $M(\rho_{AB})=m_1+m_2$, and $m_1$ and $m_2$ are the largest two eigenvalues of the matrix $\mathbb{T}^T_{\rho_{AB}}\mathbb{T}_{\rho_{AB}}$. Here, $\mathbb{T}_{\rho_{AB}}$ is the correlation matrix of the state $\rho_{AB}$.
For a two-qubit state $\rho_{AB}$ , violation of the Bell-CHSH inequality implies that the Bell-CHSH value $V(\rho_{AB})$ is greater than $2$.
In other words, a two-qubit state violates the Bell-CHSH inequality iff
\begin{equation}
M(\rho_{AB})>1.
\end{equation}
\par \textbf{Bell tripartite nonlocality:} Bell nonlocality can also be demonstrated at the tripartite level. Here a tripartite entangled state is shared between three parties Alice ($A$), Bob ($B$) and Charlie ($C$). They perform measurements on their respective part and take an account of the joint correlations. Similar to the case of the bipartite systems, if there exists a correlation which cannot be explained by a local hidden variable model, then we state that nonlocality is exhibited. However, in a pragmatic scenario, violation of a suitably chosen Bell’s inequality carries the signature of nonlocality. In our probe, we use the Svetlichny’s inequality \cite{Svetlichny,svet2} which is explained below.

Consider three parties Alice, Bob, and Charlie sharing a three-qubit quantum state $\rho_{ABC}$. Let the measurements by observers be projections onto unit vectors: $ \mathbf{A}= \vec{\sigma_1}.\vec{a} $ or $ \mathbf{A}'= \vec{\sigma_1}.\vec{a'} $ on qubit of system $A$,  $ \mathbf{B}= \vec{\sigma_2}.\vec{b} $ or $ \mathbf{B}'= \vec{\sigma_2}.\vec{b'} $ on qubit of system $B$, and $ \mathbf{C}= \vec{\sigma_2}.\vec{c} $ or $ \mathbf{C}'= \vec{\sigma_2}.\vec{c'} $ on qubit of system $C$. The Svetlichny's inequality for a three-qubit state $|\Psi_{ABC}\rangle$ is now read as
\begin{equation}\label{eq:svetlichny}
	| \langle \Psi_{ABC} | \textbf{S} | \Psi_{ABC} \rangle | \le 4,
\end{equation}
where the Svetlichny's operator is defined as 
\begin{gather}
	\textbf{S}  = (\mathbf{A} + \mathbf{A}')\otimes(\mathbf{B}\otimes \mathbf{C}' + \mathbf{B}' \otimes \mathbf{C}) \nonumber\\
	(\mathbf{A} - \mathbf{A}')\otimes(\mathbf{B}\otimes \mathbf{C} - \mathbf{B}' \otimes \mathbf{C}').
\end{gather}
The input state $|\Psi\rangle_{ABC}$ exhibits genuine tripartite nonlocality if it violates \eqref{eq:svetlichny} for all $\textbf{S}$.

\par \textbf{Quantum Steering:}  Steering, as mentioned earlier lies between entanglement and Bell nonlocality, and a state can be defined to be steerable from Alice to Bob if there is at least one probability which cannot be written as
\begin{equation}
P(a,b \vert \mathbf{A},\mathbf{B}; \rho_{AB}) = \sum_{\xi} p(\xi) p(a \vert \mathbf{A} , \xi) p(b \vert \mathbf{B} , \sigma^B_{\xi} )
\end{equation}
As we can see from the equation above, steering is an asymmetric 
property unlike Bell nonlocality and entanglement. Like Bell nonlocality and entanglement, there are several linear equalities to detect steerability. In \cite{SteeringIeq3,Costa}, authors have developed a steering inequality to check whether a bipartite state is steerable when both the parties are allowed to perform $n$  measurements on his or her part. This is given by
\begin{equation}\label{steering n}
F_n(\rho_{AB}, \{ \mathbf{A}_{i}\}_{i}, \{ \mathbf{B}_{i}\}_{i}) \coloneqq \frac{1}{\sqrt{n}} \Big|\sum_{i=1}^n\langle \mathbf{A}_i\otimes \mathbf{B}_i\rangle_{\rho_{AB}} \Big| \leq 1,
\end{equation}
where  $\langle \mathbf{A}_i \otimes \mathbf{B}_i \rangle = \text{Tr}(\rho_{AB} \mathbf{A}_i\otimes \mathbf{B}_i)$ and $ \rho_{AB}  $ is some bipartite quantum state. We say that a quantum state $\rho_{AB}$ is $n$-steerable then the following inequality holds:
\begin{equation}\label{steering n max}
    F_n^{*}(\rho_{AB}) \coloneqq \max_{\{ \mathbf{A}_{i}\}_{i}, \{ \mathbf{B}_{i}\}_{i}}\frac{1}{\sqrt{n}} \Big|\sum_{i=1}^n\langle \mathbf{A}_i\otimes \mathbf{B}_i\rangle_{\rho_{AB}} \Big| > 1.
\end{equation}
Now, if we maximize over all the observables, then it attains an optimal value $\sqrt{n}$. Therefore, we have:
\begin{equation}\label{steering n bounds}
    1 < F_n^{*}(\rho_{AB}) \leq \sqrt{n}.
\end{equation}

\subsection{Cloning} \label{subsec:cloning}
\noindent The no-cloning theorem tells us that if we are provided with an unknown quantum state $|\psi\rangle$, it is impossible to construct a complete positive trace preserving (CPTP) map $C$ which in principle  will make a perfect copy of $|\psi\rangle$ , i.e., $ C : |\psi\rangle \rightarrow  |\psi\rangle\otimes|\psi\rangle$. However, as we noted before, this never rules out the  existence of approximate quantum cloning machines. 

In this article, for our study, we consider the optimal and universal quantum cloning machine named as Buzek-Hillery (B-H) cloning machine. Mathematically, the B-H cloning machine ($U_{bh}$) is a
$M$-dimensional unitary transformation acting on a quantum state 
$|\psi_i\rangle_{a_0}$, where $i=1,\ldots ,M$. This state is to be copied on  $|0\rangle_{a_1}$, which we regard as the blank state of the machine. The initial state of the machine or copier is given by $|X\rangle_{x}$. The transformed state of the machine as a result of the cloning process is given by the set  of  state vectors  $|X_{ii}\rangle_{x}$ and $|Y_{ij}\rangle_{x}$ .  Here $a_0$, $a_1$, and $x$ are the indices used to represent the input, blank, and machine qubits, respectively. Additionally, these  transformed state vectors form an orthonormal basis of the $M$-dimensional Hilbert space.  Putting these together, the cloning transformation scheme $U_{bh}$ can be formally written as, 
\begin{eqnarray}\label{eq:bhM}
U_{bh}  |\psi_i\rangle_{a_0}|0\rangle_{a_1} |X\rangle_{x} \rightarrow c |\psi_i\rangle_{a_0} |\psi_i\rangle_{a_1}|X_{ii}\rangle_{x} +\nonumber\\
d \sum_{j \neq i}^M \left(|\psi_i\rangle_{a_0} |\psi_j\rangle_{a_1}+|\psi_j\rangle_{a_0} |\psi_i\rangle_{a_1}\right)|Y_{ij}\rangle_{x}.
\end{eqnarray}
Here, the coefficients $c$ and $d$ are complex numbers. Now, imposing the unitarity and normalization conditions on the above B-H cloner ($U_{bh}$) gives rise to the following constraints:
\begin{eqnarray}
\langle X_{ii}|X_{ii}\rangle = \langle Y_{ij}|Y_{ij}\rangle = \langle X_{ii}|Y_{ji}\rangle =1, 
\end{eqnarray}
where we have  $\langle X_{ii}|Y_{ij}\rangle = \langle Y_{ji}|Y_{ij}\rangle = \langle X_{ii}|X_{jj}\rangle =0$ with $i\neq j$ and $c^2=\frac{2}{M+1}$. Here, we consider $M=2^m$, $m$ being the number of qubits. 

In this work, we take into account the cloning of a single qubit, i.e., $M = 2$. For the sake of simplicity, we explicitly write the cloning transformations below derived from \eqref{eq:bhM}:
\begin{eqnarray}\label{eq:bhqubit}
U_{bh}  |0\rangle_{a_0}|0\rangle_{a_1} |X\rangle_{x} \rightarrow c |0\rangle_{a_0} |0\rangle_{a_1}|X_{0}\rangle_{x} +\nonumber\\
d  \big(|0\rangle_{a_0} |1\rangle_{a_1}+|1\rangle_{a_0} |0\rangle_{a_1}\big)|Y_{0}\rangle_{x} \nonumber\\
U_{bh}  |1\rangle_{a_0}|0\rangle_{a_1} |X\rangle_{x} \rightarrow c |1\rangle_{a_0} |1\rangle_{a_1}|X_{1}\rangle_{x} +\nonumber\\
d  \big(|1\rangle_{a_0} |0\rangle_{a_1}+|0\rangle_{a_0} |1\rangle_{a_1}\big)|Y_{1}\rangle_{x}.
\end{eqnarray}
In addition to the above constraints, we obtain the following constraints on the coefficients and scalar product of ancilla qubits,
\begin{gather*}\label{eq:bhcon}
   |c|^{2} + 2|d|^{2} = 1 \nonumber\\
\langle X_{0}|Y_{0}\rangle = \langle Y_{1}|Y_{0}\rangle = \langle X_{0}|X_{1}\rangle = 0.
\end{gather*}
Furthermore, we get the following expression of fidelity ($f$) of the output state with respect to the input state after the application of the above B-H quantum cloner, 
\begin{eqnarray}\label{eq:fidelityBH}
  f = |c|^2 = \frac{1}{2}\big(1 + \text{Re}(c^*d \langle X_{1}|Y_{0}\rangle + d^*c \langle Y_{1}|X_{0}\rangle) \big).
\end{eqnarray}
Let us denote $\mu = \text{Re}(c^*d \langle X_{1}|Y_{0}\rangle + d^*c \langle Y_{1}|X_{0}\rangle$). By setting both the scalar products to 1, we obtain the optimal universal B-H cloner ($U_{obh}$). On the contrary, by varying the scalar products, we obtain less optimal and more general universal B-H cloners $U_{gbh}^{\mu}$. For the optimal B-H cloner, we get $\mu = 2/3$ and $f = 5/6$. We refer to $\mu$ as a machine parameter that can be varied to get different cloners. 

In general, the density matrix of a qubit (say $\rho$) can be represented in terms of the set of $2\otimes 2$ Pauli matrices $\{\sigma_{i}\}_{i}$ where $i=\{1, 2, 3\}$, i.e., $\rho = \frac{\mathbb{I} +  \vec{x} \cdot \vec{\sigma}}{2}$. After the application of the $U_{gbh}^{\mu}$ cloning transformation, we get a two qubit output state $\rho_{out}$. As we know that symmetric cloners are better than asymmetric ones , we assume that the two qubits are in the same state i.e. $\operatorname{Tr}_{1}[\rho_{out}] = \operatorname{Tr}_{2}[\rho_{out}] = \frac{\mathbb{I} +  \mu \vec{x} \cdot \vec{\sigma}}{2}$. An important thing to note is that $\mu$ is also known as the \textit{shrinking factor} as after cloning the vector associated with Pauli matrices is shrinked by the factor of $\mu$. 

\subsubsection{Local cloning}
Let us consider we have $n$ parties $A_{1}, A_{2},\ldots,  A_{n}$ and a quantum state $\rho$ is shared between them. Additionally, each party has a qubit serving as a blank state $|0\rangle$, as well as a qubit representing the machine state $|X\rangle$ in their lab. Let us denote the joint state of $n$ blank states and $n$ machines states in the following way:
\begin{eqnarray}
E = |0\rangle \langle 0|^{\otimes n}\ \text{and}\ 
X = |X\rangle \langle X|^{\otimes n}.
\end{eqnarray}

By local cloning, we mean each party clones their share of the state $\rho$ using B-H cloning machines onto their own blank state. Formally, this cloning process can be written as follows,
\begin{eqnarray}
\tilde{\rho} = U_{A_{1}}\otimes ...\otimes U_{A_{n}} (\rho \otimes E \otimes X) U_{A_{1}}^{\dagger}\otimes...\otimes U_{A_{n}}^{{\dagger}}
\end{eqnarray}
where the state $\tilde{\rho}$ is the joint state of the input, blank, and machine states of all the parties after the application of the local B-H cloners $U_{A_{1}},\ldots, U_{A_{n}}$. Additionally, for the above process, let $\mu_{i}$ be the machine parameter associated with the cloning machine $U_{A_{i}}$. For the purpose of this work, we will use local cloning operations as it makes more sense in the case of broadcasting of resources over distances.

\section{Broadcasting of Bipartite Quantum Nonlocality}\label{broad}

\noindent The present section details the aspects of broadcasting of bipartite nonlocality and further probes the possibility of broadcasting nonlocal resources. In particular, we apply local cloning operations, as described above, on individual labs of each party sharing the resource to create more number of resources across the labs.  

 Let us consider two parties Alice (A) and Bob (B) sharing a general two qubit quantum state $\rho_{AB}$ with an additional constraint of being nonlocal (in terms of violation of Bell-CHSH inequality or the $n$-measurement steering inequality, and this will be clear from the context). Additionally, the qubit $a$ and the qubit $b$ serve as the initial blank state in Alice's and Bob's lab, respectively. Let us denote the machine qubits as $x$ and $y$ for Alice and Bob's local cloning machines, respectively.  Furthermore, we apply local cloning unitaries $U_{Aax}^{\mu_{1}}\otimes U_{Bby}^{\mu_{2}}$ on qubits ($A,a,x$) and qubits ($B,b,y$). Here, we assume that both the machine parameters, i.e., $\mu_{1}$ and $\mu_{2}$, are different for generalizing the approach. Now, first, tracing out ancilla qubits $x,y$ on Alice's and Bob's side respectively, we get the output state as $\tilde{\rho}_{ABab}$. We trace out the ($B, b$) and ($A, a$) qubits to obtain the output states $\tilde{\rho}_{Aa}$ on Alice’s side and $\tilde{\rho}_{Bb}$ on Bob’s side, respectively. Similarly, after tracing out appropriate qubits from the output state, we obtain the two plausible groups of output states $\tilde{\rho}_{Ab}$ and $\tilde{\rho}_{aB}$ across Alice's and Bob's laboratory. 
\par The expression for distant-lab output states (note that we call them distant-lab because they are in different labs) $\tilde{\rho}_{Ab}$ and $\tilde{\rho}_{aB}$ across Alice's and Bob's labs are given as
\begin{eqnarray}
\begin{aligned}
&\tilde{\rho}_{Ab}  \\ & = \operatorname{Tr}_{aBxy}\left [U_{Aax}^{\mu_{1}}\otimes U_{Bby}^{\mu_{2}}(\rho_{AB}\otimes E_{ab}\otimes X_{xy})U_{Aax}^{\mu_{1}\dagger}\otimes U_{Bby}^{\mu_{2}\dagger}\right] \\ \\
& \tilde{\rho}_{aB} \\ & = \operatorname{Tr}_{Abxy}\left [U_{Aax}^{\mu_{1}}\otimes U_{Bby}^{\mu_{2}}(\rho_{AB}\otimes E_{ab}\otimes X_{xy})U_{Aax}^{\mu_{1}\dagger}\otimes U_{Bby}^{\mu_{2}\dagger}\right].
\end{aligned}
\end{eqnarray}
 The expression for same-lab output states within Alice's and Bob's labs are given as
\begin{eqnarray}
\begin{aligned}
&\tilde{\rho}_{Aa} \\ & = \operatorname{Tr}_{Bbxy}\left [U_{Aax}^{\mu_{1}}\otimes U_{Bby}^{\mu_{2}}(\rho_{AB}\otimes E_{ab}\otimes X_{xy})U_{Aax}^{\mu_{1}\dagger}\otimes U_{Bby}^{\mu_{2}\dagger}\right] \\ \\
&\tilde{\rho}_{Bb} \\ & = \operatorname{Tr}_{Aaxy}\left [U_{Aax}^{\mu_{1}}\otimes U_{Bby}^{\mu_{2}}(\rho_{AB}\otimes E_{ab}\otimes X_{xy})U_{Aax}^{\mu_{1}\dagger}\otimes U_{Bby}^{\mu_{2}\dagger}\right].
\end{aligned}
\end{eqnarray}
Here, $E_{ab} $ and $X_{xy}$ represent the initial blank state and machine state, respectively. 

\begin{definition}
Given a nonlocal input state $\rho_{AB}$, we say that nonlocality is \textit{broadcasted} if after the application of local B-H quantum cloning operations $U_{Aax}^{\mu_{1}}\otimes U_{Bby}^{\mu_{2}}$, the distant-lab output states, i.e., $\tilde{\rho}_{Ab}$ and $\tilde{\rho}_{aB}$, are nonlocal with respect to the concerned inequalities (i.e., Bell or steering inequalities).
\end{definition}

Now, if we consider the two-qubit general mixed state between Alice and Bob $\rho_{AB}$ as shown in \eqref{eq:genmixedstate}, then Lemma \ref{lem:2qubitclonedoutput} describes the output states that we get after the application of the local B-H quantum cloners. Formally, we write the lemma in the following way.

\begin{lemma}\label{lem:2qubitclonedoutput}
Given a general two-qubit mixed state $\rho_{AB} = \{\vec{x}, \vec{y}, \mathbb{T}\}$, the final distant-lab output states after the application of local B-H cloners ($U_{Aax}^{\mu_{1}}\otimes U_{Bby}^{\mu_{2}}$) are $\tilde{\rho}_{Ab} = \tilde{\rho}_{aB} = \{\mu_{1}\vec{x}, \mu_{2}\vec{y}, \mu_{1}\mu_{2}\mathbb{T}\}$. 
\end{lemma}
\begin{proof}
The proof of this lemma is provided in Appendix~\ref{Appendix A}.
\end{proof}

It is important to note that for symmetric cloners like  Buzek-Hillery cloners outputs $\tilde{\rho}_{Ab}$ and $\tilde{\rho}_{aB}$ are identical. Interestingly, in this work we find that it is impossible to broadcast bipartite nonlocality if we restrict ourselves to general local B-H cloning machines. 

\subsection{Broadcasting of Bell Nonlocal States}

\noindent  We start with a general two-qubit state $\rho_{AB}$, shared between Alice and Bob, which violates the Bell-CHSH inequality, i.e., $ 1 < M(\rho_{AB}) \leq 2$.
In what follows, we show that it is impossible to broadcast Bell nonlocality, considered in the purview of Bell-CHSH inequality.  

\begin{theorem}
	Bell nonlocality of an arbitrary quantum state $\rho_{AB}  = \{ \vec{x}, \vec{y}, \mathbb{T}\}$ cannot be broadcasted using general B-H local quantum cloners.
\end{theorem}
\begin{proof}
	The input quantum state $\rho_{AB}$ violates the Bell-CHSH inequality:
	\begin{equation}
	1 < M(\rho_{AB}) \leq 2.
	\label{eq:CHSHrho}
	\end{equation}
	Here, $M(\rho_{AB})$ is the sum of the two largest eigenvalues ($\eta_{1},\ \eta_{2}$) of $\mathbb{T}^{T}\mathbb{T}$. Therefore, we have,
	\begin{equation}
	1 < \eta_{1} + \eta_{2} \leq 2.
	\label{eq:CHSHeig}
	\end{equation}
	According to Lemma \ref{lem:2qubitclonedoutput}, the output state $\tilde{\rho}_{Ab}$ (or equivalently $\tilde{\rho}_{aB}$), obtained after the application of local B-H quantum cloners with machine parameters $\mu_{1}$ and $\mu_{2}$, is given as,
	\begin{equation}
	\tilde{\rho}_{Ab} = \left\{  \mu_{1}\vec{x}, \mu_{2}\vec{y}, \mu_{1}\mu_{2}\mathbb{T} \right\}.
	\end{equation}
	Now, the two largest eigenvalues of $\mu_{1}^{2}\mu_{2}^{2}\mathbb{T}^{T}\mathbb{T}$ are $(\eta_{1}', \eta_{2}') \coloneqq (\mu_{1}^{2}\mu_{2}^{2}\eta_{1},\ \mu_{1}^{2}\mu_{2}^{2}\eta_{2})$. From \eqref{eq:CHSHeig}, we get
	\begin{gather}
	1 < \eta_{1} + \eta_{2} \leq 2 \nonumber\\
	\mu_{1}^{2}\mu_{2}^{2} < \mu_{1}^{2}\mu_{2}^{2}\eta_{1} + \mu_{1}^{2}\mu_{2}^{2} \eta_{2} < 2 \mu_{1}^{2}\mu_{2}^{2}\nonumber\\
	\mu_{1}^{2}\mu_{2}^{2} < \eta_{1}' + \eta_{2}' \leq 2 \mu_{1}^{2}\mu_{2}^{2}.
	\end{gather}

 However, the maximum value that the machine parameters $\mu_{1}$ and $\mu_{2}$ can take is $2/3$ which corresponds to the machine parameter of the optimal universal B-H cloner (see Section \ref{subsec:cloning}). Hence, the sum of eigenvalues is restricted by
	\begin{gather}
	\frac{1}{4} \leq \eta_{1}' + \eta_{2}' < \frac{1}{2} \nonumber\\
	M(\tilde{\rho}_{Ab}) \leq 1/2.
	\end{gather}
	
 This proves that bipartite Bell nonlocality  cannot broadcast by any BH local cloners.
\end{proof}

\subsection{Broadcasting of $n$-Steerable States}
\noindent In this section, we study the broadcasting of $n$-steerability with $n > 2$. For the case of $n=2$, steerability is equivalent to Bell nonlocality. We assume that the initial state $\rho_{AB}$ is $n$-steerable to begin with. 
\begin{theorem}\label{steering}
	The $n$-steerability of an arbitrary quantum state $\rho_{AB}  = \{ \vec{x}, \vec{y}, \mathbb{T}\}$ can be broadcasted using BH local quantum cloners if $n > 5$.
\end{theorem}
\begin{proof}

As the initial state $\rho_{AB}$ is $n$-steerable, the following holds according to \eqref{steering n bounds},
\begin{equation}\label{eq:steerrho}
1 < F_{n}^{*}(\rho_{AB}) \leq \sqrt{n}.
\end{equation}
According to Lemma \ref{lem:2qubitclonedoutput}, the output state $\tilde{\rho}_{Ab}$ (or equivalently $\tilde{\rho}_{aB}$), after the application of local B-H quantum cloners, is given as,
	\begin{equation}
	\tilde{\rho}_{Ab} = \bigg\{  \mu_{1}\vec{x}, \mu_{2}\vec{y}, \mu_{1}\mu_{2}\mathbb{T} \bigg\}.
	\end{equation} 
	
Next, we verify if $\tilde{\rho}_{Ab}$ satisfies \eqref{steering n bounds}. First, we can write
\begin{align}
  F_{n}^{*}(\tilde{\rho}_{Ab}) & = \max_{\{ \mathbf{A}_{i}\}_{i}, \{ \mathbf{B}_{i}\}_{i}} \frac{1}{\sqrt{n}} \Bigg| \sum_{i=1}^{n}\langle \mathbf{A}_{i} \otimes \mathbf{B}_{i} \rangle_{\tilde{\rho}_{Ab}} \Bigg|  \nonumber \\
  & = \max_{\{ \mathbf{A}_{i}\}_{i}, \{ \mathbf{B}_{i}\}_{i}} \frac{1}{\sqrt{n}} \Bigg| \sum_{i=1}^{n} \operatorname{Tr}[\tilde{\rho}_{Ab} \mathbf{A}_{i} \otimes \mathbf{B}_{i}] \Bigg|\label{steering out1}.
\end{align}
Let $\mathbf{A}_{i} = \Hat{e}_{i}\cdot\vec{\sigma}$ and $\mathbf{B}_{i} = \Hat{f}_{i}\cdot\vec{\sigma}$, where the vector $\vec{\sigma} = (\sigma_{1}, \sigma_{2}, \sigma_{3})^{\mathsf{T}}$ is composed of Pauli matrices. Additionally, $\Hat{e}_{i} \in \mathbb{R}^{3}$ are unit vectors and $\Hat{f}_{i} \in \mathbb{R}^{3}$ are orthogonal vectors. Therefore, we can rewrite \eqref{steering out1} in the following way:
\begin{align}
 & F_{n}^{*}(\tilde{\rho}_{Ab}) \notag \\ & = \max_{\{ \Hat{e}_{i}\}_{i}, \{ \Hat{f}_{i}\}_{i}} \frac{1}{\sqrt{n}} \Bigg| \sum_{i=1}^{n} \operatorname{Tr}\big[\tilde{\rho}_{Ab} \Hat{e}_{i}\cdot\vec{\sigma} \otimes \Hat{f}_{i}\cdot\vec{\sigma}\big] \Bigg|\nonumber \\
  & = \max_{\{ \Hat{e}_{i}\}_{i}, \{ \Hat{f}_{i}\}_{i}} \frac{1}{\sqrt{n}} \Bigg| \sum_{i=1}^{n} \operatorname{Tr}\Big[\tilde{\rho}_{Ab} \sum_{r=1}^{3} e_{ir} \sigma_{r} \otimes \sum_{r'=1}^{3} f_{ir'} \sigma_{r'}\Big] \Bigg| \nonumber \\
  & = \max_{\{ \Hat{e}_{i}\}_{i}, \{ \Hat{f}_{i}\}_{i}}\frac{1}{\sqrt{n}} \Bigg| \sum_{i=1}^{n} \operatorname{Tr}\Big[\tilde{\rho}_{Ab} \sum_{r,r'=1}^{3} e_{ir} f_{ir'} \sigma_{r} \otimes \sigma_{r'}\Big] \Bigg|\nonumber \\
  & = \max_{\{ \Hat{e}_{i}\}_{i}, \{ \Hat{f}_{i}\}_{i}} \frac{1}{\sqrt{n}} \Bigg| \sum_{i=1}^{n} \sum_{r,r'=1}^{3} e_{ir} f_{ir'} \operatorname{Tr}\Big[\tilde{\rho}_{Ab}  \sigma_{r} \otimes \sigma_{r'}\Big] \Bigg|\nonumber \\
  & = \max_{\{ \Hat{e}_{i}\}_{i}, \{ \Hat{f}_{i}\}_{i}} \frac{1}{\sqrt{n}} \Bigg| \sum_{i=1}^{n} \sum_{r,r'=1}^{3} e_{ir} f_{ir'} \mu_{1}\mu_{2}\operatorname{Tr}\Big[\rho_{AB}  \sigma_{r} \otimes \sigma_{r'}\Big] \Bigg|.
\end{align}
The last equality is due to Lemma \ref{lem:2qubitclonedoutput}. Upon further simplification, we have, 
\begin{align}
    & F^{*}_{n}(\tilde{\rho}_{Ab}) \notag \\ & = \max_{\{ \Hat{e}_{i}\}_{i}, \{ \Hat{f}_{i}\}_{i}} \frac{\mu_{1}\mu_{2}}{\sqrt{n}} \Bigg| \sum_{i=1}^{n} \sum_{r,r'=1}^{3} e_{ir} f_{ir'} \operatorname{Tr}\Big[\rho_{AB} \sigma_{r} \otimes \sigma_{r'}\Big] \Bigg|\nonumber \\
     & =  \max_{\{ \Hat{e}_{i}\}_{i}, \{ \Hat{f}_{i}\}_{i}}\frac{\mu_{1}\mu_{2}}{\sqrt{n}} \Bigg| \sum_{i=1}^{n} \operatorname{Tr}\Big[\rho_{AB} \sum_{r=1}^{3} e_{ir} \sigma_{r} \otimes \sum_{r'=1}^{3} f_{ir'} \sigma_{r'}\Big]\Bigg| \nonumber \\
     & = \max_{\{ \Hat{e}_{i}\}_{i}, \{ \Hat{f}_{i}\}_{i}} \frac{\mu_{1}\mu_{2}}{\sqrt{n}} \Bigg| \sum_{i=1}^{n} \operatorname{Tr}\big[\rho_{AB} \Hat{e}_{i}\cdot\vec{\sigma} \otimes \Hat{f}_{i}\cdot\vec{\sigma}\big] \Bigg| \nonumber \\
     & = \max_{\{ A_{i}\}_{i}, \{ B_{i}\}_{i}} \frac{\mu_{1}\mu_{2}}{\sqrt{n}} \Bigg| \sum_{i=1}^{n} \operatorname{Tr}\big[\rho_{AB} A_{i} \otimes B_{i}\big] \Bigg| \nonumber \\
     & = \mu_{1}\mu_{2} F_{n}^{*}(\rho_{AB}).
\end{align}
Now, according to \eqref{eq:steerrho}, we have
\begin{equation}
    \mu_{1}\mu_{2} < F^{*}_{n}(\tilde{\rho}_{Ab}) \leq \mu_{1}\mu_{2} \sqrt{n}.
\end{equation}

We cannot broadcast $n$-steering for all the $n$-steerable input states $\rho_{AB}$ because the lower bound $\mu_{1}\mu_{2} < 1; \forall \mu_{1}, \mu_{2} \in [0, 2/3] $. In contrast, for broadcasting $n$-steering for some input states, the upper bound should be greater than 1. Therefore, we have
\begin{gather}
    \mu_{1}\mu_{2} \sqrt{n} > 1, \nonumber\\
    n > \frac{1}{(\mu_{1}\mu_{2})^{2}}.
\end{gather}
Furthermore, we substitute $\mu_{1} = \mu_{2} = 2/3$ for obtaining the best lower bound of $n$. These correspond to the machine parameters of the optimal universal B-H quantum cloners. Hence,
\begin{equation}
    n > \frac{81}{16} = 5.0625.
\end{equation}
In conclusion, for $n \ge 6$, some of the output states are $n$-steerable after the application of local optimal B-H cloners.
\end{proof}

Theorem \ref{steering} shows that for certain number of measurements, the corresponding steering inequality is violated by the output states for some input states $\rho_{AB}$. Now, we take two types of two-qubit pure states and show that they become unsteerable under any number of measurement settings, i.e., they admit a local hidden state (LHS) model [Theorem \ref{thm:werner} and \ref{thm:belldiagonal}]. In the work \cite{Bowles}, the authors had derived a sufficient criterion for the unsteerability of a two-qubit state based on its bloch parameters. States which satisfy the criteria admits a LHS model, thereby rendering them unsteerable under any number of measurement settings. In the following theorems, we show that in some instances when a steerable Werner state or Bell diagonal state is used for broadcasting of steerability, states, which admit LHS, are obtained as outputs.  
\begin{theorem}\label{thm:werner}
	The application of B-H local cloner on a steerable Werner state makes it unsteerable.
\end{theorem}
\begin{proof}
	The bloch representation of Werner states is given as
	\begin{equation}
	\rho_{AB} = \bigg\{  \vec{0}, \vec{0}, \mathbb{T} \bigg\}.
	\end{equation}
	where $\mathbb{T} = \text{diag}(p , -p, p)$ and $p \in [0, 1]$ is the visibility factor of the Werner state. After the application of B-H local cloning machine, the distant-lab output state is given as,
	\begin{equation}
	\tilde{\rho}_{Ab} = \bigg\{  \vec{0}, \vec{0},\mu_{1}\mu_{2}\mathbb{T} \bigg\}.
	\end{equation}
	Now the criterion in \cite{Bowles} can also be written as,
	\begin{equation}\label{eq:unsteercrit}
	|\vec{x}|^{2} + 2\sqrt{\eta_\text{max}} \leq 1,
	\end{equation}
 where $ \eta_{\text{max}} $ is the largest eigenvalue of $ \mathbb{T}^\dagger \mathbb{T} $. Now, substituting the bloch vectors of Werner states into the above equation, we get
	\begin{equation}
	\begin{split}
	2\sqrt{\eta_{\text{max}}} \leq 1 \\
	2 \mu_{1}\mu_{2} p \leq 1 \\
	p \leq \frac{1}{2\mu_{1}\mu_{2}}.
	\end{split}
	\end{equation}
	
 For maximum value of $\mu_{1}$ and $\mu_{2}$, i.e., $\frac{2}{3}$, we get $p\leq1$. For other lower values of $\mu$, we get the same range of $p$ for which it is satisfying the criterion. This means that the whole range will satisfy the unsteerability criterion after local B-H cloning. This concludes the proof. 
\end{proof}

\begin{theorem}\label{thm:belldiagonal}
	The application of B-H local cloner on a steerable Bell-diagonal state makes it unsteerable.
\end{theorem}
\begin{proof}
	Bell-diagonal states are given as
	\begin{equation}
	\rho_{AB} = \bigg\{  \vec{0}, \vec{0}, \mathbb{T} \bigg\}.
	\end{equation}
	where $\mathbb{T} = \text{diag}(c_{1} , c_{2}, c_{3})$ and $-1 \leq c_{i} \leq 1$ for all $i \in \{1, 2, 3\}$. After the application of BH local cloning machine, the output states are 
	\begin{equation}
	\tilde{\rho}_{Ab} = \bigg\{  \vec{0}, \vec{0},\mu_{1}\mu_{2}\mathbb{T} \bigg\}.
	\end{equation}
	
 Substituting the bloch vectors of the above output state (assuming the first case to be $c_{1} \geq c_{2} \geq c_{3}$ without the loss of generality) into \eqref{eq:unsteercrit}, we obtain
	\begin{equation}
	\begin{split}
	2\sqrt{\eta_{\text{max}}} \leq 1,\\
	2 \mu_{1}\mu_{2} c_{1} \leq 1,\\
	c_{1} \leq \frac{1}{2\mu_{1}\mu_{2}}.
	\end{split}
	\end{equation}
	For maximum value of $\mu_{1}$ and $\mu_{2}$, i.e., $\frac{1}{\sqrt{2}}$, we get $c_{1}\leq1$. For other lower values of $\mu$, we get the same range of $c_{1}$ for which it is satisfying the criterion. Similar line of reasoning can be used for other two cases as well, i.e., for the cases where $c_{2} \geq c_{3} \geq c_{1}$ and $c_{3} \geq c_{1} \geq c_{2}$.  This implies that the whole range will satisfy the unsteerability criterion after local B-H cloning. This concludes the proof.
\end{proof}

\section{Broadcasting of Tripartite Quantum Nonlocality}\label{broadtri}
In this section, we show the details of broadcasting of genuine tripartite nonlocality under the ambit of Svetlichny's inequality. We consider the application of B-H local cloning transformation as shown by \eqref{eq:bhqubit} on each party individually.
Let us consider three distant parties Alice, Bob, and Charlie sharing a general three-qubit quantum state $\rho_{ABC}$. Similar to the two-qubit case, the qubits $a, b$, and $c$ serve as the initial blank state in Alice's, Bob's and Charlie's lab, respectively. we denote the machine qubits as $x, y$, and $z$ for Alice's, Bob's and Charlie's local cloning machines, respectively. They apply their corresponding local cloning transformations, i.e., $U_{Aax}^{\mu_{1}}\otimes U_{Bby}^{\mu_{2}}\otimes U_{Ccz}^{\mu_{3}}$ on qubits ($A,a,x$), qubits ($B,b,y$), and qubits ($C,c,z$), respectively. One of the output states $\tilde{\rho}_{Abc}$ is obtained as,
\begin{gather}
\begin{aligned}
\tilde{\rho}_{Abc} = Tr_{aBCxyz}\Big[&U_{Aax}^{\mu_{1}}\otimes U_{Bby}^{\mu_{2}}\otimes U_{Ccz}^{\mu_{3}}(\rho_{ABC}\otimes E_{abc}\otimes X_{xyz})\\ 
&U_{Aax}^{\mu_{1}\dagger} \otimes U_{Bby}^{\mu_{2}\dagger}\otimes U_{Ccz}^{\mu_{3}\dagger}\Big],
\end{aligned}
\end{gather}
where $E_{abc}$ and $X_{xyz}$ denotes the initial blank and machine states, respectively. The other output states can be obtained in a similar manner.
\begin{definition}
Given a genuine nonlocal input state $\rho_{ABC}$, we say that nonlocality is braodcasted if after the application of local B-H quantum cloning operation $U_{Aax}^{\mu_{1}}\otimes U_{Bby}^{\mu_{2}}\otimes U_{Ccz}^{\mu_{3}}$, the distant-lab output states are genuinely nonlocal with respect to Svetlichny's inequality.
\end{definition}

Now if we consider a three-qubit quantum state between three distant parties as shown in \eqref{eq:genmixedstate3}, then Lemma \ref{lem:3qubitclonedoutput} shows one of the distant-lab output states that we get after the local B-H cloning operations. Also note that as we are using symmetric B-H local cloners, all the distant-lab output states, that obtained after the cloning process, are the same.

\begin{lemma}\label{lem:3qubitclonedoutput}
Given a general three-qubit mixed state $\rho_{ABC} = \{\vec{x},\:\vec{y},\:\vec{z},\:\mathbb{T},\:\mathbb{V},\: \mathbb{W}, \:\mathbb{G}\}$, one of the final distant-lab output states after the application of local BH cloners ($U_{Aax}^{\mu_{1}}\otimes U_{Bby}^{\mu_{2}}\otimes U_{Ccz}^{\mu_{3}}$) is $\tilde{\rho}_{Abc} = \{\mu_{1}\vec{x},\:\mu_{2}\vec{y},\:\mu_{3}\vec{z},\:\mu_{2}\mu_{3}\mathbb{T},\:\mu_{1}\mu_{3}\mathbb{V},\: \mu_{1}\mu_{2}\mathbb{W}, \:\mu_{1}\mu_{2}\mu_{3}\mathbb{G}\}$.
\end{lemma}
\begin{proof}
We do not explicitly prove this because the proof follows a similar line of reasoning used for the proof of Lemma \ref{lem:2qubitclonedoutput}.
\end{proof}

First, we begin with an assumption that $\rho_{ABC}$ showcases genuine tripartite nonlocality, i.e, it violates Svetlichny inequality. In the following theorem, we prove the impossibility of broadcasting genuine tripartite nonlocality.

\begin{theorem}
Genuine tripartite nonlocality of an arbitrary quantum state $\rho_{ABC}$ cannot be broadcasted using B-H local quantum cloners.
\end{theorem}
\begin{proof}
As the initial state showcases genuine tripartite nonlocality, it violates Svetlichny inequality
\begin{equation}\label{initialstateSI}
 4 < S(\rho_{ABC}) \leq 4\sqrt{2}.
\end{equation}

 Now, we check the violation of Svetlichny inequality $S(\cdot)$ for one of the distant-lab output states $\tilde{\rho}_{Abc}$ obtained after the application of local B-H quantum cloners.
\begin{gather}\label{eq:SIneqAbc}
    S(\tilde{\rho}_{Abc}) = \max_{S}\operatorname{Tr}[\tilde{\rho}_{Abc}\textbf{S}].
\end{gather}
where the operator $S$ can be represented in the following way:
\begin{gather*}
    \textbf{S}  = (\mathbf{A} + \mathbf{A}')\otimes(\mathbf{B}\otimes \mathbf{C}' + \mathbf{B}' \otimes \mathbf{C}) \nonumber\\
     (\mathbf{A} - \mathbf{A}')\otimes(\mathbf{B}\otimes \mathbf{C} - \mathbf{B}' \otimes C').
\end{gather*}
Here, let $\mathbf{A} = \hat{d}\cdot \vec{\sigma}$ and $\mathbf{A}' = \hat{d}'\cdot \vec{\sigma}$ for Alice's system. Similarly we have $\mathbf{B} = \hat{e}\cdot \vec{\sigma}$ and $\mathbf{B}' = \hat{e}'\cdot \vec{\sigma}$ for Bob's system, and $\mathbf{C} = \hat{f}\cdot \vec{\sigma}$ and $\mathbf{C}' = \hat{f}'\cdot \vec{\sigma}$ for Charlie's system. Additionally, let $\Hat{m} = (m_{1}, m_{2}, m_{3})$ be a unit vector in $\mathbb{R}$, where $m\in\{\hat{d}, \hat{d}', \hat{e}, \hat{e}', \hat{f}, \hat{f}' \}$. Substituting in \eqref{eq:SIneqAbc}, we get
\begin{multline}
 S(\tilde{\rho}_{Abc}) = \max_{\substack{\mathbf{A}, \mathbf{A}', \mathbf{B},\\ \mathbf{B}', \mathbf{C} , \mathbf{C}'}} \operatorname{Tr}\Bigg[\tilde{\rho}_{Abc}\Big((\mathbf{A} + \mathbf{A}')\otimes(\mathbf{B}\otimes \mathbf{C}' + \mathbf{B}' \otimes \mathbf{C})\\
     (\mathbf{A} - \mathbf{A}')\otimes(\mathbf{B}\otimes \mathbf{C} - \mathbf{B}' \otimes \mathbf{C}')\Big)\Bigg].
\end{multline}
We simplify the above equation as,
\begin{widetext}
\begin{align}
S(\tilde{\rho}_{Abc}) & = \max_{\hat{d}, \hat{d}', \hat{e}, \hat{e}', \hat{f}, \hat{f}'} \operatorname{Tr}\Bigg[\tilde{\rho}_{Abc} \Big( \sum_{i,j,k=1}^{3} (d_{i} + d'_{i})(e_{j} f'_{k} + e'_{j} f_{k}) 
 + (d_{i} - d'_{i})(e_{j} f_{k} - e'_{j} f'_{k}) (\sigma_{i}\otimes \sigma_{j}\otimes \sigma_{k})  \Big)\Bigg]\\
 & = \max_{\hat{d}, \hat{d}', \hat{e}, \hat{e}', \hat{f}, \hat{f}'} \sum_{i,j,k=1}^{3} (d_{i} + d'_{i})(e_{j} f'_{k} + e'_{j} f_{k}) + (d_{i} - d'_{i}) (e_{j} f_{k} - e'_{j} f'_{k}) \operatorname{Tr}\Big[\tilde{\rho}_{Abc}(\sigma_{i}\otimes \sigma_{j}\otimes \sigma_{k}) \Big]\\
 & = \max_{\hat{d}, \hat{d}', \hat{e}, \hat{e}', \hat{f}, \hat{f}'} \sum_{i,j,k=1}^{3} (d_{i} + d'_{i})(e_{j} f'_{k} + e'_{j} f_{k})
 + (d_{i} - d'_{i}) (e_{j} f_{k} - e'_{j} f'_{k}) \mu_{1}\mu_{2}\mu_{3}\operatorname{Tr}\Big[\rho_{ABC}(\sigma_{i}\otimes \sigma_{j}\otimes \sigma_{k}) \Big]\\
 & = \mu_{1}\mu_{2}\mu_{3} \max_{\mathbf{A}, \mathbf{A}', \mathbf{B},\mathbf{ B}', \mathbf{C} , \mathbf{C}'} \operatorname{Tr}\Bigg[\rho_{ABC}\Big((\mathbf{A} + \mathbf{A}')\otimes(\mathbf{B}\otimes \mathbf{C}' + \mathbf{B}' \otimes \mathbf{C}) (\mathbf{A} - \mathbf{A}')\otimes(\mathbf{B}\otimes \mathbf{C} - \mathbf{B}' \otimes \mathbf{C}')\Big)\Bigg] .
\end{align}
\end{widetext}
The above equation can be rewritten in terms of $S(\rho_{ABC})$ as,
\begin{eqnarray}
S(\tilde{\rho}_{Abc}) = \mu_{1}\mu_{2}\mu_{3} S(\rho_{ABC}).
\end{eqnarray}
From Eq. \ref{initialstateSI}, we get
\begin{gather}
4 \mu_{1}\mu_{2}\mu_{3} \leq S(\tilde{\rho}_{Abc}) \leq 4\sqrt{2}\mu_{1}\mu_{2}\mu_{3},\\
32/27 \leq S(\tilde{\rho}_{Abc}) \leq 32\sqrt{2}/27 < 4.
\end{gather}
This means that the output state $\tilde{\rho}_{Abc}$ is not genuinely nonlocal. To conclude, we prove the impossibility of broadcasting of genuine tripartite nonlocality using B-H local quantum cloners.
\end{proof}

\section{broadcasting using general unitaries}\label{broadgen}
In this section, we try to broadcast Bell nonlocality and steering using  general arbitrary local unitaries for two-qubit systems, instead of restricting to local cloning unitaries.  We ran numerical simulation of $10^5$ unitaries generated using the Haar measure over Werner and Bell-diagonal states. 

For the setup, we have two parties, i.e., Alice and Bob, sharing a nonlocal state $\rho_{AB}$. They both apply an arbitrary local unitary ($U_{al}$) on their respective qubits. The nonlocal output states after the unitary operation are given as
\begin{equation}
\begin{split}
& \rho_{Ab}=\operatorname{Tr}_{Ba}[U_{al}\otimes U_{al} (\rho_{AB}\otimes \sigma_{ab})U_{al}^{\dagger}\otimes U_{al}^{\dagger}],\\
& \rho_{aB} =\operatorname{Tr}_{Ab}[U_{al}\otimes U_{al} (\rho_{AB}\otimes \sigma_{ab})U_{al}^{\dagger}\otimes U_{al}^{\dagger}],
\end{split}    
\end{equation}
where $\sigma_{ab} \coloneqq \dyad{00}$. We find that, for none of these states, broadcasting of nonlocality, as well as  3-steerability is possible.

\section{Conclusion}\label{conc}
\noindent Nonlocality is an extremely important ingredient in quantum information processing. As a consequence, broadcasting nonlocal resources from a few assumes considerable significance.
In this article, we show that unlike entanglement it is impossible to broadcast nonlocality as a resource for the choice of  cloning machine as the Buzek-Hillery transformation and applying them locally. We study in the purview of the Bell-CHSH inequality and CJWR steering inequality to exhibit that nonlocal resources cannot be broadcasted. The loss can also be to the extent of the state being rendered as unsteerable under any number of measurement settings. Though this is not counter intuitive but different from entanglement. Our result not only limits to local cloning operations but also extends for $10^5$ simulated local unitaries for two qubit systems. This enables us to conjecture that nonlocality exhibited in the form of violation of the Bell-CHSH inequality and CJWR steering inequality can not be broadcasted even with arbitrary local unitaries.

\section*{Acknowledgement}
Nirman Ganguly acknowledges support from the project grant received from DST-SERB (India) under the MATRICS scheme, vide file number MTR/2022/000101.

\section*{Data Availability Statement}
No data associated in the manuscript.

\onecolumngrid
\begin{appendix} 
	\section{Proof of Lemma~\ref{lem:2qubitclonedoutput}} \label{Appendix A}
	Let $\omega_{A} = \frac{\mathbb{I} + \vec{x}\cdot \vec{\sigma}}{2} \in \mathbb{C}^{2}$ be a state of quantum system $A$. Also, let the blank state and the machine state be denoted as $E_{a}$ and $X_{x}$, respectively. Let $U^{\mu}_{Aax}$ be a unitary operator applied on the joint system of $A$, $a$, and $x$, and it defines the BH cloning machine with the machine parameter $\mu$. After applying this cloning transformation to the joint system $\omega_{A} \otimes E_{a} \otimes X_{x}$, we get the following output states $\omega^{out}_{A}$ and $\omega^{out}_{a}$:
	\begin{align}
		\label{eq:omegaoutA}
		\omega^{out}_{A} & = \operatorname{Tr}_{ax}\left [U^{\mu}_{Aax} \left ( \omega_{A} \otimes E_{a} \otimes X_{x} \right)U^{\mu \dagger}_{Aax}\right ]
		= \frac{\mathbb{I} + \mu \vec{x}\cdot \vec{\sigma}}{2}\\
		\omega^{out}_{a} & = \operatorname{Tr}_{Ax}\left [U^{\mu}_{Aax} \left ( \omega_{A} \otimes E_{a} \otimes X_{x} \right)U^{\mu \dagger}_{Aax}\right ]
		= \frac{\mathbb{I} + \mu \vec{x}\cdot \vec{\sigma}}{2}\label{eq:omegaouta}
	\end{align}
	Both the output states are the same as we are using a symmetric cloning machine. Intuitively, when we compare these output states with the input state $\omega$, it is very interesting to observe that the vector associated with Pauli matrices is shrinked by the factor of $\mu$. This observation is cruicial for this proof.
	
	Now, let us expand this development to the two-qubit case. As given, Alice (A) and Bob (B) share a state $\rho_{AB}$, which is a general two-qubit state as written in \eqref{eq:genmixedstate}. 
	Additionally, let $E_{a}$ and $E_{b}$ serve as initial blank quantum states in Alice's and Bob's lab, respectively. 
	Let us denote the machine qubits as $x$ and $y$ for Alice and Bob's local cloning machines, respectively.  
	Furthermore, they apply local cloning unitaries $U_{Aax}^{\mu_{1}}\otimes U_{Bby}^{\mu_{2}}$ on their respective qubits, i.e., qubits ($A,a,x$) and qubits ($B,b,y$). 
	Here, we assume that both the machine parameters, i.e., $\mu_{1}$ and $\mu_{2}$, need not be same for generalizing the approach. Let $\tilde{\rho}_{Ab}$ and $\tilde{\rho}_{Ba}$ be nonlocal output states obtained after cloning. Here, we only show for $\tilde{\rho}_{Ab}$ as we are using symmetric cloners, and therefore, we have $\tilde{\rho}_{Ab} = \tilde{\rho}_{Ba}$.

	\begin{align}
		\tilde{\rho}_{Ab} & = \operatorname{Tr}_{aBxy}\left [U_{Aax}^{\mu_{1}}\otimes U_{Bby}^{\mu_{2}}\left(\rho_{AB}\otimes E_{a}\otimes E_{b} \otimes X_{x} \otimes X_{y} \right)U_{Aax}^{\mu_{1}\dagger}\otimes U_{Bby}^{\mu_{2}\dagger}\right] \nonumber\\
		& \overset{(a)}{=} \operatorname{Tr}_{aBxy}\Bigg [U_{Aax}^{\mu_{1}}\otimes U_{Bby}^{\mu_{2}}\Bigg(\frac{1}{4}\left[\mathbb{I}_A \otimes \mathbb{I}_B +\sum_{i=1}^{3}\big(x_{i}(\sigma^{i}_{A}\otimes \mathbb{I}_B)+ y_{i}(\mathbb{I}_A\otimes\sigma^{i}_{B})\big)+\sum_{i,j=1}^{3}t_{ij}(\sigma^{i}_{A} \otimes\sigma^{j}_{B})\right]\nonumber\\
        & \hspace{10cm} \otimes E_{a}\otimes E_{b} \otimes X_{x} \otimes X_{y} \Bigg) U_{Aax}^{\mu_{1}\dagger}\otimes U_{Bby}^{\mu_{2}\dagger}\Bigg]\\
		& \overset{(b)}{=} 1/4 \Bigg (\operatorname{Tr}_{aBxy}\left [U_{Aax}^{\mu_{1}}\otimes U_{Bby}^{\mu_{2}}\left(\mathbb{I}_A \otimes E_{a} \otimes X_{x} \otimes \mathbb{I}_B \otimes E_{b} \otimes X_{y} \right )U_{Aax}^{\mu_{1}\dagger}\otimes U_{Bby}^{\mu_{2}\dagger} \right] \\ 
		& \hspace{3cm} +\sum_{i=1}^{3} x_{i} \operatorname{Tr}_{aBxy}\left [U_{Aax}^{\mu_{1}}\otimes U_{Bby}^{\mu_{2}} \left(\sigma^{i}_{A}\otimes E_{a} \otimes X_{x} \otimes \mathbb{I}_B \otimes E_{b} \otimes X_{y} \right)U_{Aax}^{\mu_{1}\dagger}\otimes U_{Bby}^{\mu_{2}\dagger} \right]\\
		& \hspace{3cm} + \sum_{i=1}^{3} y_{i} \operatorname{Tr}_{aBxy}\left [U_{Aax}^{\mu_{1}}\otimes U_{Bby}^{\mu_{2}} \left(\mathbb{I}_{A}\otimes E_{a} \otimes X_{x} \otimes \sigma^{i}_B \otimes E_{b} \otimes X_{y} \right)U_{Aax}^{\mu_{1}\dagger}\otimes U_{Bby}^{\mu_{2}\dagger} \right] \\
		& \hspace{3cm} +\sum_{i, j=1}^{3} t_{ij} \operatorname{Tr}_{aBxy}\left [U_{Aax}^{\mu_{1}}\otimes U_{Bby}^{\mu_{2}} \left(\sigma^{i}_{A}\otimes E_{a} \otimes X_{x} \otimes \sigma^{j}_B \otimes E_{b} \otimes X_{y} \right)U_{Aax}^{\mu_{1}\dagger}\otimes U_{Bby}^{\mu_{2}\dagger} \right]\Bigg)\\
		& \overset{(c)}{=} 1/4 \Bigg (\operatorname{Tr}_{ax}\left [U_{Aax}^{\mu_{1}}\left(\mathbb{I}_A \otimes E_{a} \otimes X_{x}\right)U_{Aax}^{\mu_{1}\dagger} \right] \otimes \operatorname{Tr}_{By}\left [U_{Bby}^{\mu_{2}}\left(\mathbb{I}_B \otimes E_{b} \otimes X_{y}\right)U_{Bby}^{\mu_{2}\dagger} \right] \\ 
		& \hspace{3cm} +\sum_{i=1}^{3} x_{i} \operatorname{Tr}_{ax}\left [U_{Aax}^{\mu_{1}}\left(\sigma^{i}_A \otimes E_{a} \otimes X_{x}\right)U_{Aax}^{\mu_{1}\dagger} \right] \otimes \operatorname{Tr}_{By}\left [U_{Bby}^{\mu_{2}}\left(\mathbb{I}_B \otimes E_{b} \otimes X_{y}\right)U_{Bby}^{\mu_{2}\dagger} \right]\\
		& \hspace{3cm} + \sum_{i=1}^{3} y_{i} \operatorname{Tr}_{ax}\left [U_{Aax}^{\mu_{1}}\left(\mathbb{I}_A \otimes E_{a} \otimes X_{x}\right)U_{Aax}^{\mu_{1}\dagger} \right] \otimes \operatorname{Tr}_{By}\left [U_{Bby}^{\mu_{2}}\left(\sigma^{i}_B \otimes E_{b} \otimes X_{y}\right)U_{Bby}^{\mu_{2}\dagger} \right] \\
		& \hspace{3cm} +\sum_{i, j=1}^{3} t_{ij} \operatorname{Tr}_{ax}\left [U_{Aax}^{\mu_{1}}\left(\sigma^{i}_A \otimes E_{a} \otimes X_{x}\right)U_{Aax}^{\mu_{1}\dagger} \right] \otimes \operatorname{Tr}_{By}\left [U_{Bby}^{\mu_{2}}\left(\sigma^{j}_B \otimes E_{b} \otimes X_{y}\right)U_{Bby}^{\mu_{2}\dagger} \right]\Bigg)\\
		& \overset{(d)}{=} 1/4\left( \mathbb{I}_{A} \otimes \mathbb{I}_{b} + \sum_{i=1}^{3} x_{i} \mu_{1} (\sigma^{i}_{A} \otimes \mathbb{I}_{b}) + \sum_{i=1}^{3} y_{i} \mu_{2} (\mathbb{I}_{A} \otimes \sigma^{i}_{b}) + \sum_{i, j=1}^{3} t_{ij} \mu_{1} \mu_{2} (\sigma^{i}_{A} \otimes \sigma^{j}_{b}) \right)\\
		& = \{\mu_{1} \vec{x}, \mu_{2} \vec{y}, \mu_{1} \mu_{2}\mathbb{T} \}.
	\end{align}
	The equality (a) follows from the definition of a general two-qubit mixed state, the equalities (b) and (c) follow by expanding the expression and rearranging it for simplicity, and the equality (d) follows from \eqref{eq:omegaoutA} and \eqref{eq:omegaouta}.

\end{appendix}

\end{document}